\documentclass[11pt]{article}
\def\withcolors{0}
\def\withnotes{0}
\usepackage[T1]{fontenc}
\usepackage[utf8]{inputenc}

\usepackage{lmodern}
\usepackage{xspace}                                     

\usepackage[normalem]{ulem}

\usepackage{amsfonts,amsmath,amssymb, amsthm, mathtools}
\usepackage{thm-restate}

\usepackage{algorithmicx,algpseudocode,algorithm}

\usepackage[usenames,dvipsnames,table]{xcolor}

\usepackage{csquotes}

\usepackage{relsize}

\usepackage{multirow}
\usepackage{chngpage} 

\usepackage{mfirstuc}

\usepackage{tikz}
\usetikzlibrary{arrows}
\usetikzlibrary{calc,decorations.pathmorphing,patterns}

\usepackage[backref,colorlinks,citecolor=blue,bookmarks=true,linktocpage]{hyperref}
\usepackage{aliascnt}
\usepackage[numbered]{bookmark}

\usepackage{fullpage}

\usepackage{titling}

\usepackage[shortlabels]{enumitem}
  \setitemize{noitemsep,topsep=3pt,parsep=2pt,partopsep=2pt} 
  \setenumerate{itemsep=1pt,topsep=2pt,parsep=2pt,partopsep=2pt}
  \setdescription{itemsep=1pt}
  
\ifnum\withnotes=1
  \usepackage[colorinlistoftodos,textsize=scriptsize]{todonotes}
\fi

\usepackage{verbatim}
\usepackage{bigfoot} 

\usepackage{mleftright} 

\usepackage{pgfplots}

\input{preamble}

\newcommand{\p}{\mathbf{p}}
\newcommand{\q}{\mathbf{q}}
\newcommand{\unif}{\mathbf{u}}

\newcommand{\infnorm}[1]{\norminf{#1}}
\newcommand{\twonorm}[1]{\normtwo{#1}}
\newcommand{\twonorms}[1]{\normtwo{#1}^2}
\newcommand{\threenorm}[1]{\norm{#1}_3}
\newcommand{\threenormc}[1]{\norm{#1}_3^3}

\newcommand{\U}{\class_U}

\newcommand{\Ex}{\shortexpect}

\def\authornametb{Tu\u{g}kan Batu}
\def\authoraffitb{London School of Economics. Email: \email{t.batu@lse.ac.uk}.}
\def\authornamecc{Cl\'ement L. Canonne}
\def\authorafficc{Columbia University. Email: \email{ccanonne@cs.columbia.edu}. Research supported by NSF grants CCF-1115703 and NSF CCF-1319788.}

\title{Generalized Uniformity Testing}
\date{\today}
 
\author{
  \tcolor{\authornametb}\thanks{\authoraffitb}
  \and \ccolor{\authornamecc}\thanks{\authorafficc}
}

\begin{document}

\maketitle

\pagenumbering{gobble}
\begin{abstract}
    In this work, we revisit the problem of \emph{uniformity testing} of
discrete probability distributions. A fundamental problem in distribution testing, testing uniformity over a \emph{known} domain has been addressed over a significant line of works, and is by now fully understood.

The complexity of deciding whether an unknown distribution is uniform over its unknown (and arbitrary) \emph{support}, however, is much less clear. Yet, this task arises as soon as no prior knowledge on the domain is available, or whenever the samples originate from an unknown and unstructured universe. In this work, we introduce and study this \emph{generalized} uniformity testing question, and establish nearly tight upper and lower bound showing that -- quite surprisingly -- its sample complexity significantly differs from the known-domain case. Moreover, our algorithm is intrinsically \emph{adaptive}, in contrast to the overwhelming majority of known distribution testing algorithms.

\end{abstract}

\ifnum\withnotes=1
  \clearpage
  \listoftodos
  \hfill
  \newpage
  \tableofcontents
  \clearpage
\fi

\clearpage\pagenumbering{arabic}

\section{Introduction}\label{sec:intro}
Property testing, as introduced in the seminal works of~\cite{RS:96,GGR:98}, is the analysis and study of ultra-efficient and randomized decision algorithms, which must answer a promise problem yet cannot afford to query their whole input. A very successful and prolific area of theoretical computer science, property testing also gave rise to several subfields, notably that of \emph{distribution} testing, where the input consists of independent \emph{samples} from a probability distribution, and one must now verify if the underlying unknown distribution satisfies a given property of interest (cf.~\cite{Ron:08,Ron:09,Rubinfeld:12:Survey,Canonne:15:Survey, Gol:17} for surveys on property and distribution testing).

One of the earliest and most studied questions in distribution testing
is that of \emph{uniformity testing}, where, given independent samples
from an arbitrary probability distribution $\p$ on a discrete domain
$\domain$, one has to decide whether (i) $\p$ is uniform on $\domain$,
or (ii) $\p$ is ``far'' (i.e., at total variation distance at least
$\eps$) from the uniform distribution on $\domain$. Arguably the most
natural distribution testing problem, testing uniformity is also one of the most fundamental; algorithms for uniformity testing end up being crucial building blocks in many other distribution testing algorithms~\cite{BFFKRW:01,DK:16,Gol:16}. Fully understanding the sample complexity of the problem, as well as the possible trade-offs it entails, thus prompted a significant line of research.

Starting with the work of Goldreich and Ron~\cite{GRexp:00}
(which considered it in the context of testing expansion of graphs), uniformity testing was studied and analyzed in a series of work~\cite{BFFKRW:01,Paninski:08,VV:14,DKN:15,ADK:15,DGPP:16}, which culminated with the tight sample complexity bound of $\bigTheta{\sqrt{n}/\eps^2}$ for testing uniformity on a discrete domain of size $n$. (Moreover, the corresponding algorithms are also efficient, running in time linear in the number of samples they take.)

Given this state of affairs, testing uniformity of discrete
distributions appears to be fully settled; however, as often is the case, the devil is in the detail. Specifically, all the aforementioned results address the case where the domain $\domain$ is explicitly known, and the task is to find out whether $\p$ is the uniform distribution \emph{on this domain}. Yet, in many cases, samples (or data points) are drawn from the underlying distribution without such prior knowledge, and the relevant question is whether $\p$ is uniform on its \emph{support} -- which is unknown, of arbitrary size, and can be completely unstructured.\footnote{In particular, one cannot without loss of generality assume that the support is the set of consecutive integers $\{1,\dots,n\}$.}

In this work, we focus on this latter question: in particular, we do not assume any \textit{a priori} knowledge on the domain $\domain$, besides its being discrete. Our goal is then the following: given independent samples from an arbitrary probability distribution $\p$ on $\domain$, we must distinguish between the case (i) $\p$ is uniform on \emph{some subset of $\domain$}, and (ii) $\p$ is far from \emph{every} such uniform distribution. As we shall see, this is not merely a technicality: this new task is provably harder than the case where $\domain$ is known. Indeed, this difference intuitively stems from the uncertainty on where the support of $\p$ lies, which prevents any reduction to the simple, known-domain case.

Furthermore, one crucial feature of the problem is that it intrinsically calls for \emph{adaptive} algorithms. This is in sharp contrast to the overwhelming majority of distribution testing algorithms, which (essentially) draw a prespecified number of samples all at once, before processing them and outputting a verdict. This is because, in our case, an algorithm is provided only with the proximity parameter $\eps\in(0,1]$, and has no upper bound on the domain size $n$ nor on any other parameter of the problem. Therefore, it must keep on taking samples until it has ``extracted'' enough information -- and is confident enough that it can stop and output an answer. (In this sense, our setting is closer in spirit to the line of work pioneered in Statistics by Ingster~\cite{Ingster:2000,Fromont:2006} than to the ``instance-optimal'' setting of Valiant and Valiant~\cite{VV:14,BCG:16}, as in the latter the algorithm is still provided with a massive parameter in the form of the full description of a reference probability distribution.)

\subsection{Our Results}

Given a discrete, possibly unbounded domain $\domain$, we let $\class_U$ denote the set of all probability distributions that are supported and uniform on some subset of $\domain$, that is
\[
  \class_U \eqdef \setOfSuchThat{ \unif_S }{S\subseteq \domain}
\]
where, for a given set $S\subseteq\domain$, $\unif_S$ denote the uniform distribution on $S$. In what follows, we write $\totalvardist{\p}{\q}$ for the total variation distance between two distributions $\p,\q$ on $\domain$.

\begin{restatable}{theorem}{theomainub}\label{theo:main:theorem:ub}
  There exists an algorithm which, given sample access to an arbitrary distribution $\p$ over some unknown discrete domain $\domain$, as well as parameter $\eps\in(0,1]$, satisfies the following.
  \begin{enumerate}
    \item If $\p\in\class_U$, then the algorithm outputs \accept with probability at least $2/3$; while
    \item if $\totalvardist{\p}{\class_U}>\eps$, then the algorithm outputs \reject with probability at least $2/3$.
  \end{enumerate}
  Moreover, the algorithm takes $\bigO{\frac{1}{\eps^6\norm{\p}_3}}$ samples in expectation, and is efficient (in the number of samples taken).
\end{restatable}

We note that if indeed $\p$ is uniform, i.e., $\p=\unif_S$ for some $S\subseteq\domain$, then, for constant $\eps$, the above complexity becomes $O\big(\abs{S}^{2/3}\big)$ -- to be compared to the $\bigTheta{\sqrt{\abs{S}}}$ sample complexity of testing whether $\p=\unif_S$ for a fixed~$S$. Our next result shows that this is not an artifact of our algorithm; namely, such a dependence is necessary, and testing the \emph{class} of uniform distributions is strictly harder than testing any specific uniform distribution.

\begin{restatable}{theorem}{theomainlb}\label{theo:main:theorem:lb}
  Fix any (non-uniform) distribution $\q$ over $\domain$, and let
  $\eps\eqdef \totalvardist{\q}{\class_U}$ be its distance to
  $\class_U$. Then, given sample access to a distribution $\p$ on
  $\domain$, distinguishing with high constant probability between (i) $\p$ is equal to $\q$ up to a
  permutation of the domain and (ii) $\p\in\class_U$, requires
  $\bigOmega{\frac{1}{\norm{\q}_3}}$ samples. In particular, an
  algorithm that tests membership in $\class_U$ with high probability
  and for any proximity parameter $\eps'\leq \eps$ requires this many samples. 
\end{restatable}
It is worth discussing the above statement in detail, as its interpretation can be slightly confusing. Specifically, it does \emph{not} state that testing identity to any fixed, known distribution $\p$ requires $\bigOmega{{1}/{\norm{\p}_3}}$ (indeed, by the results of~\cite{VV:14,BCG:16}, such a statement would be false). What is stated is essentially that, even given the full description of $\p$, it is hard to distinguish between $\p$ and a uniform distribution, \emph{after relabeling of the elements of the domain}. Since the class of uniform distributions is invariant by such permutations, the last part of the theorem follows.


\subsection{Overview and Techniques}

The key intuition and driving idea of both our upper and lower bounds is the observation that, by very definition of the problem, there is no structure nor ordering of the domain to leverage. That is, the class of uniform distributions over $\domain$ is a ``symmetric property'' (broadly speaking, the actual labeling of the elements of the domain is irrelevant), and the domain itself can and should be thought of as a set of arbitrary points with no algebraic structure. Given this state of affairs, an algorithm should not be able to do much more than counting \emph{collisions}, that is the number of pairs, or triples, or more generally $k$-tuples of samples which happen to ``hit'' the same domain element.

Equivalently, these collision counts correspond to the \emph{moments} (that is, $\lp[p]$-norms) of the distribution; following a line of works on symmetric properties of distributions (\cite{GRexp:00,RRSS:09,Valiant:11,ValiantValiant:11}, to cite a few), we thus need to, and can only, focus on estimating these moments. To relate this to our property $\class_U$, we first need a simple connection between $\lp[p]$ norms and uniformity of a distribution. However, while getting an exact characterization is not difficult (\autoref{lemma:equal:moments:characterize:uniform}), we are interested in a \emph{robust} characterization, in order to derive a correspondence between approximate equality between $\lp[p]$ norms and distance to uniformity. This is what we obtain in~\autoref{lem:smallnormtouni}: roughly speaking, if $\normtwo{\p}^4\approx\norm{\p}_3^3$ then $\p$ must be close to a uniform distribution on $1/\normtwo{\p}^2$ elements.

This in turn allows us to design and analyze a simple and clean testing algorithm, which works in two stages: (i) estimate $\normtwo{\p}^2$ to sufficient accuracy; (ii) using this estimate, take enough samples to estimate $\norm{\p}_3^3$ as well; and \accept if and only if $\normtwo{\p}^4\approx\norm{\p}_3^3$.

Turning to the lower bound, the idea is once again to only use the available information: namely, if all that \emph{should} matter are the $\lp[p]$-norms of the distribution, then two distributions with similar low-order norms \emph{should} be hard to distinguish; so it would suffice to come up with a pair of uniform and far-from-uniform distributions $\p^\yes,\p^\no$ with similar moments to establish our lower bound. Fortunately, this intuition~--~already present in~\cite{RRSS:09}~--~was formalized and developed in an earlier work of Paul Valiant~\cite{Valiant:11}, which we thus can leverage for our purpose. Given this ``Wishful Thinking Theorem'' (see~\autoref{theo:valiant:wishful}), what remains is to upper bound the discrepancy of the moments of our two candidate distributions $\p^\yes,\p^\no$ to show that some specific quantity is very small. Luckily, this last step also can be derived from the aforementioned robust characterization,~\autoref{lem:smallnormtouni}.

\subsection{Organization}

After recalling some useful notation and results in~\autoref{sec:preliminaries}, we establish our upper bound (\autoref{theo:main:theorem:ub}) in~\autoref{sec:upperbound}. \autoref{sec:lowerbound} is then dedicated to the proof of our lower bound, \autoref{theo:main:theorem:lb}.

\section{Preliminaries}\label{sec:preliminaries}
\subsection{Definitions and notation}
All throughout this paper, we write $\distribs{\domain}$ for the set of discrete probability distributions over domain $\domain$, i.e. the set of all real-valued functions $\p\colon\domain\to[0,1]$ such that $\sum_{x\in\domain} \p(x) =1$. Considering a probability distribution as the vector of its probability mass function (pmf), we write $\norm{\p}_r$ for its $\lp[r]$-norm, for any $r\in[1,\infty]$. A \emph{property} of distributions over $\domain$ is then a subset $\property\subseteq \distribs{\domain}$, comprising all distributions that have the property.

As standard in distribution testing, we will measure the distance between two distributions $\p_1, \p_2$ on $\domain$ by their \emph{total variation distance}
\[
\totalvardist{\p_1}{\p_2 } \eqdef \frac{1}{2} \normone{\p_1 - \p_2} = \max_{S \subseteq \domain} (\p_1(S)-\p_2(S))
\]
which takes value in $[0,1]$. (This metric is sometimes referred to as \emph{statistical distance}). Given a property $\property$ and a distribution $\p\subseteq\distribs{\domain}$, we then write $\totalvardist{\p}{\property}\eqdef \inf_{\q\in\property} \totalvardist{\p}{\q}$ for the distance of $\p$ to $\property$.

Finally, recall that a \emph{testing algorithm} for a fixed property $\property$ is a randomized algorithm $\Tester$ which takes as input a proximity parameter $\eps\in(0,1]$, and is granted access to independent samples from an unknown distribution $\p$:
\begin{enumerate}
\item if $\p \in \property$, the algorithm outputs \accept with probability at least $2/3$;
\item if $\totalvardist{\p}{\p^\prime} \geq \eps$ for every $\p^\prime\in\property$, it outputs \reject with probability at least $2/3$.
\end{enumerate}
That is, \Tester must accept if the unknown distribution has the property, and reject if it is \emph{\eps-far} from having it. The  \emph{sample complexity} of the algorithm is the number of samples it draws from the distribution in the worst case.

\subsection{Useful results from previous work}

We will heavily rely, for our lower bound, on the ``Wishful Thinking Theorem'' due to Paul Valiant~\cite{Valiant:11}, which applies to testing symmetric properties of distributions (that is, properties that are invariant under relabeling of the domain, as $\class_U$ happens to be). Intuitively, this theorem ensures that ``if the low-degree moments ($\lp[p]$ norms) of two distributions match, then these distributions (up to relabeling) are hard to distinguish.''

\begin{theorem}[Wishful Thinking Theorem {\cite[Theorem 4.10]{Valiant:11}, restated}]\label{theo:valiant:wishful}
  Given a positive integer $k$ and two distributions $\p^\yes,\p^\no$, it is impossible to test in $k$ samples any symmetric property that holds for $\p^\yes$ and does not hold for $\p^\no$, provided that following conditions hold:
  \begin{itemize}
    \item $\norminf{\p^\yes},\norminf{\p^\no} \leq \frac{1}{500k}$;
    \item letting $m^\yes$, $m^\no$ be the $k$-based moments of $\p^\yes,\p^\no$ (defined below),
      \[
         \sum_{j=2}^\infty \frac{\abs{m^{\yes}(j)-m^{\no}(j)}}{\sqrt{1+\max(m^{\yes}(j),m^{\no}(j))}} < \frac{1}{24},
      \]
      where $m^{\yes}(j) \eqdef k^j\norm{\p^\yes}_j^j$, $m^{\no}(j) \eqdef k^j\norm{\p^\no}_j^j$ for $j\geq 0$.
  \end{itemize}
\end{theorem}
\noindent (We observe that we only reproduced here one of the three sufficient conditions given in the original, more general theorem; as this will be the only one we need.)

\subsection{Some structural results}

We here state and establish some simple yet useful results. The first relates uniformity of a distribution to the $\lp[p]$-norms of its probability mass function, while the second provides inequalities between these norms.

\begin{lemma}\label{lemma:equal:moments:characterize:uniform}
  Let $\p\in\distribs{\domain}$. Then, $\normtwo{\p}^4=\norm{\p}_3^3$ if and only if $\p\in\class_U$.
\end{lemma}
\begin{proof}
  If $\p\in\class_U$, it is immediate to see that $\normtwo{\p}^4=\norm{\p}_3^3$. We thus consider the converse implication. By the Cauchy--Schwarz inequality, 
  \[
      \normtwo{\p}^2 = \sum_{i\in\domain} \p_i^2 
      \leq \left( \sum_{i\in\domain} \left(\p_i^{3/2}\right)^2 \right)^{1/2}\left( \sum_{i\in\domain} \left(\p_i^{1/2}\right)^2 \right)^{1/2}
      =  \left( \sum_{i\in\domain} \p_i^3 \right)^{1/2}\left( \sum_{i\in\domain} \p_i \right)^{1/2} = \norm{\p}_3^{3/2}\cdot 1
  \]
  with equality if, and only if, $(\p_i^{3/2})_{i\in\domain}$ and $(\p_i^{1/2})_{i\in\domain}$ are linearly dependent. Thus, $\normtwo{\p}^4=\norm{\p}_3^3$ implies that there exist non-zero $\alpha,\beta\in\R$ such that $\alpha \p_i^{3/2} = \beta \p_i^{1/2}$ for all $i\in\Omega$, or equivalently that $\p_i \in \{0,\frac{\beta}{\alpha}\}$ for all $i\in\domain$. This, in turn, implies that $\p$ is uniform on a subset of $\frac{\alpha}{\beta}$ elements.
\end{proof}

  \begin{fact}\label{fact:ineq:moments:holder}
    For any vector $x\in\R^{\N}$ such that $\normone{x} <\infty$, we have
    \[
        \normtwo{x}^{2(j-1)} \leq \normone{x}^{j-2}\norm{x}_j^j,
    \]
    for all $j\geq 2$. In particular, for any distribution $\p\in\distribs{\domain}$, we have  $\normtwo{\p}^{2(j-1)} \leq \norm{\p}_j^j$ for all $j\geq 2$ (and, thus, for instance, 
    $\normtwo{\p}^{4} \leq \norm{\p}_3^3$).
  \end{fact}
  \begin{proof}
  The inequality is trivially true for $j=2$, and, so, we henceforth assume $j\geq 3$. Let $x\in\R^{\N}$ be such a vector: we wish to show that $\left(\sum_{i=0}^\infty x_i^2\right)^{j-1} \leq \left(\sum_{i=0}^\infty \abs{x_i}\right)^{j-2}\left(\sum_{i=0}^\infty \abs{x_i}^j\right)$, or equivalently $\sum_{i=0}^\infty x_i^2 \leq \left(\sum_{i=0}^\infty \abs{x_i}\right)^\frac{j-2}{j-1}\left(\sum_{i=0}^\infty \abs{x_i}^j\right)^\frac{1}{j-1}$. Set $p'\eqdef \frac{j-1}{j-2}$, and $q'\eqdef j-1$ so that $p',q'\geq 1$ with $\frac{1}{p'}+\frac{1}{q'}=1$. Observing that $\abs{x_i}^2=\abs{x_i}^{\frac{j-2}{j-1}}\abs{x_i}^{\frac{j}{j-1}}$, we then apply H\"older's inequality:
\begin{align*}
      \sum_{i=0}^\infty \abs{x_i}^2 & = \sum_{i=0}^\infty\abs{x_i}^\frac{1}{p'}\abs{x_i}^\frac{j}{q'}\\
  &\leq \left(\sum_{i=0}^\infty \abs{x_i}^\frac{p'}{p'}\right)^\frac{1}{p'}\left(\sum_{i=0}^\infty \abs{x_i}^\frac{jq'}{q'}\right)^\frac{1}{q'}\\
      &= \left(\sum_{i=0}^\infty \abs{x_i}\right)^\frac{j-2}{j-1}\left(\sum_{i=0}^\infty \abs{x_i}^j\right)^\frac{1}{j-1}
    \end{align*}
concluding the proof.
  \end{proof}

\section{The Upper Bound}\label{sec:upperbound}
Our algorithm for testing uniformity first estimates the $\lp[2]$
norm of the input distribution and uses this estimate to obtain a
surrogate value for the size of the support set for the
distribution. In the case the input distribution is a uniform
distribution, the $\lp[2]$ norm estimate indeed provides a good
approximation to the size of the support set. Our algorithm for the
$\lp[2]$ norm estimation is presented in the following section,
followed by our algorithm for testing uniformity.

\subsection{Estimating the $\lp[2]$ norm of a distribution}
In this section, we present an algorithm that, given independent
samples from a distribution $\p$ over $\N$, estimates
$\twonorms{\p}$. Note that a similar result was presented in Batu et
al.~\cite{BFRSW:13} in the case when the size of the domain is bounded
and known to the algorithm. Furthermore, an algorithm based on the
same ideas have been presented by Batu et al.~\cite{BDKR:05} to
estimate the entropy of a distribution that is uniform on a subset of
its domain. The algorithm is presented below in~\autoref{alg:twonorm}.

\begin{algorithm}[H]
  \caption{Estimating the $\lp[2]$ norm of a distribution from samples}
  \label{alg:twonorm}
  \begin{algorithmic}[1]
  \Procedure{Estimate-$\lp[2]$-norm}{$\p,\eps$}
  \State $k\gets \lceil \frac{C}{\eps^4}\rceil$ \Comment{$C=6500$}
  \State Keep taking samples from $\p$ until $k$ $2$-collisions are
  observed.
  \State Let $m$ be the number of samples taken.
  \State \Return $\frac{k}{\binom{m}2}$
  \EndProcedure  
  \end{algorithmic}
\end{algorithm}

\begin{lemma} \label{lem:normest}
  Algorithm~\textsc{Estimate-$\lp[2]$-norm}, given independent samples
  from a distribution $\p$ over $\N$ and $0<\eps<\frac12$, outputs
  a value $\gamma$ such that
\begin{equation}\label{eq:l2:estimate:guarantee}
 (1-\eps)\cdot\twonorms{\p} \leq \gamma \leq (1+\eps)\cdot
\twonorms{\p},
\end{equation}
with probability at least $3/4$.  Whenever the algorithm produces an
estimate satisfying~\eqref{eq:l2:estimate:guarantee} above, the number of samples taken by the algorithm is
$\Theta(\frac{1}{\eps^2\twonorm{\p}})$. Moreover, the algorithm takes $O(\frac{1}{\eps^2\twonorm{\p}})$ samples in expectation.
\end{lemma}
\begin{proof}
Let $M$ be the random variable that denotes the number of samples that
were taken by the algorithm until $k$ pairwise collisions are observed. We will
show that, with constant probability, $M$ is close to its expected
value nearly $\sqrt{k}/\twonorm{\p}$.

Consider a set of $m$ samples from $\p$. For $1\le i < j \le m$, let
$X_{ij}$ be an indicator random variable denoting a collision between
$i$th and $j$th samples. Let $S_m=\sum_{1\le i < j \le m} X_{ij}$ be
the total number of collisions among the samples.    

For any $i<j$, $\Ex[X_{ij}]=\twonorms{\p}$. Therefore,
$\Ex[S_m]=\binom{m}{2}\cdot\twonorms{\p}$. We will also need an upper
bound on the variance $\var[S_m]$ to show that the $k$ collisions are
not observed too early or too late.
\[
 \Ex\left[ S_m^2 \right] = \Ex\left[ \left(\sum_{i<j}X_{ij}\right)\left(\sum_{i'<j'}X_{i'j'}\right)\right] = \sum_{i<j, i'<j'} \Ex[X_{ij}X_{i'j'}].
\]
The terms of the last summation above can be grouped according to the
cardinality of the set $\{i,j,i',j'\}$.
\begin{itemize}
\item If $|\{i,j,i',j'\}|=2$, then
  $\Ex[X_{ij}X_{i'j'}]=\Ex[X_{ij}]=\twonorms{\p}$. There are
  $\binom{m}{2}$ such terms.
\item If $|\{i,j,i',j'\}|=3$, then 
  $\Ex[X_{ij}X_{i'j'}]=\Ex[X_{ij}X_{ij'}]=\threenormc{\p}$. There are
$6\binom{m}{3}$ such terms.
\item If $|\{i,j,i',j'\}|=4$, then
  $\Ex[X_{ij}X_{i'j'}]=\Ex[X_{ij}]\Ex[X_{i'j'}]=\twonorm{\p}^4$. There are
  $6\binom{m}{4}$ such terms.
\end{itemize}
Hence, we can bound the variance of $S_m$ as follows.
\begin{align*}
\var[S_m] &= \Ex[S_m^2] - \Ex[S_m]^2 \\
&= \binom{m}{2} \cdot\twonorms{\p} + 6\binom{m}{3}\cdot\threenormc{\p} +
6\binom{m}{4}\cdot \twonorm{\p}^4 - \left( \binom{m}{2}
  \cdot\twonorms{\p}\right)^2  \\
&= \binom{m}{2} \cdot\twonorms{\p} + 2m\cdot \threenormc{\p} +
  (m^3-3m^2)\cdot (\threenormc{\p} - \twonorm{\p}^4)\\
&\le \binom{m}{2} \cdot\twonorms{\p} + m^3 \cdot \twonorm{\p}^3,
\end{align*}
where the inequality arises from $\threenorm{\p}\le\twonorm{\p}$.

The probability that the output of the algorithm is less than
$ (1-\eps)\cdot\twonorms{\p}$ (that is, an underestimation) is
bounded from above by the probability of the random variable $M$
taking a value $m$ such that
$(1-\eps)\binom{m}{2}\cdot \twonorms{\p} > k$. Analogously, the
probability of an overestimation is bounded above by the probability
of the random variable $M$ taking a value $m$ such that
$(1+\eps)\binom{m}{2}\cdot \twonorms{\p} < k$.

Let $m$ be the smallest integer such that
$(1+\eps)\binom{m}{2}\cdot \twonorms{\p} \geq k$, so that $(1+\eps)\binom{m-1}{2}\cdot \twonorms{\p} < k$. Then,
\begin{align*}
  \Pr\left[\text{overestimation}\right] &= \Pr[M< m] = \Pr[ \exists \ell \leq m-1,\  S_\ell \geq k]
  = \Pr\left[S_{m-1} \geq k\right] \\
&= \Pr\left[S_{m-1} - \Ex[S_{m-1}] \geq k - \binom{m-1}{2}\cdot \twonorms{\p}\right] \\
&\le \Pr\left[\left| S_{m-1} - \Ex[S_{m-1}]\right| > \eps\cdot\binom{m-1}{2}\cdot  \twonorms{\p}\right]\\
&\le \frac{\var[S_{m-1}]}{\left(\eps\cdot\binom{m-1}{2}\cdot  \twonorms{\p} \right)^2} \tag{Chebyshev's inequality}\\
&\le \frac{\binom{m-1}{2} \cdot\twonorms{\p} + (m-1)^3 \cdot \twonorm{\p}^3}{\eps^2 \cdot\binom{m-1}{2}^2\cdot \twonorm{\p}^4}\\
&\le\frac{1}{\eps^2}\left(\frac{1}{\binom{m-1}{2}\cdot \twonorms{\p}}+ \frac{9}{(m-1)\cdot \twonorm{\p}} \right)\\
&=\frac{1}{\eps^2}\left(\frac{m}{m-2}\frac{1}{\binom{m}{2}\cdot \twonorms{\p}}+ \frac{m}{m-1}\frac{9}{m\cdot \twonorm{\p}} \right)\\
&\leq\frac{1}{\eps^2}\left(\frac{10}{8}\frac{1}{\binom{m}{2}\cdot \twonorms{\p}}+ \frac{10}{9}\frac{9}{m\cdot \twonorm{\p}} \right) \tag{$m\geq \sqrt{2k}+1\geq 10$, or $\Pr[S_{m-1} \geq k]=0$.}\\
&\operatorname*{\le}_{(\ast)} \frac{1}{\eps^2}\left(\frac{10}{8}\cdot \frac{1+\eps}{k} + \frac{10\sqrt{1+\eps}}{\sqrt{2k}}\right)  
\le \frac{10}{\eps^2}\left(\frac{1}{4k} + \frac{1}{\sqrt{k}}\right) \\ 
&\le \frac{5\eps^2}{2C} + \frac{10}{\sqrt{C}} \leq \frac{5}{2C} + \frac{10}{\sqrt{C}}\\
&<\frac{1}{8}
\end{align*}
for $C\geq 6500$, where $(\ast)$ follows from the choice of $m$.

To upper bound the probability of underestimation, take $m$ to be
largest integer such that 
$(1-\eps)\binom{m}{2}\cdot \twonorms{\p} \leq k$ (so that $(1-\eps)\binom{m+1}{2}\cdot \twonorms{\p} > k$, i.e. $(1-\eps)\frac{m+1}{m-1}\binom{m}{2}\cdot \twonorms{\p} > k$).\footnote{In particular, this implies $\binom{m+1}{2}> k$, from which $m > \sqrt{2k}+1 \gg \frac{1}{\eps^2}$.} Then,
\begin{align*}
  \Pr\left[\text{underestimation}\right] &= \Pr[M> m] = \Pr[ \forall \ell \leq m,\  S_\ell < k]
  = \Pr\left[S_{m} < k\right] \\
&= \Pr\left[\Ex[S_{m}] - S_{m} > \Ex[S_{m}] - k\right] \\
&\le \Pr\left[\Ex[S_{m}] - S_{m} > \left(1-(1-\eps)\frac{m+1}{m-1}\right) \cdot\binom{m}{2}\cdot  \twonorms{\p}\right]\\
&= \Pr\left[\Ex[S_{m}] - S_{m} > \frac{\eps m -1 }{m-1} \cdot\binom{m}{2}\cdot  \twonorms{\p}\right] \tag{Note that $\eps m > 1$}\\
&\le \left(\frac{m-1}{\eps m -1 }\right)^2\frac{\var[S_{m}]}{\left(\binom{m}{2}\cdot  \twonorms{\p} \right)^2}
\le \left(\frac{2}{\eps}\right)^2\frac{\var[S_{m}]}{\left(\binom{m}{2}\cdot  \twonorms{\p} \right)^2}\\
&\le 4\frac{\binom{m}{2} \cdot\twonorms{\p} + m^3 \cdot \twonorm{\p}^3}{\eps^2 \cdot\binom{m}{2}^2\cdot \twonorm{\p}^4}
\le\frac{4}{\eps^2}\left(\frac{1}{\binom{m}{2}\cdot \twonorms{\p}}+ \frac{9}{m\cdot \twonorm{\p}} \right)\\
&\leq\frac{4}{\eps^2}\left(\frac{12}{10}\frac{1}{\binom{m}{2}\cdot \twonorms{\p}}+ \frac{11}{10}\frac{9}{m\cdot \twonorm{\p}} \right) \tag{$m\geq \sqrt{2k}+1\geq 10$.}\\
&\operatorname*{\le}_{(\ast)} \frac{6}{\eps^2}\left(\frac{1-\eps}{k} + \frac{\sqrt{1-\eps}}{\sqrt{2k}}\right)  
\le \frac{6}{\eps^2}\left(\frac{1}{k} + \frac{1}{\sqrt{2k}}\right) \\ 
&\le \frac{6\eps^2}{C} + \frac{6}{\sqrt{2C}} \leq \frac{6}{C} + \frac{6}{\sqrt{2C}}\\
&<\frac{1}{8}
\end{align*}
for $C\geq 1250$, where $(\ast)$ follows from the choice of $m$.\cmargin{To do: doublecheck changes I've made.}

By the union bound, overestimation or underestimation happens with
probability at most 1/4. Finally, in the event that we have a good
estimation, we have that the number $m$ of samples satisfy
\[
 \frac{k}{(1+\eps)\cdot\twonorms{\p}} \le \binom{m}{2} \le \frac{k}{(1-\eps)\cdot\twonorms{\p}}.
\]
Therefore, we have that
$m=\Theta(\sqrt{k}/\twonorm{\p})=
\Theta(1/(\eps^2\cdot\twonorm{\p}))$. \medskip

\noindent To bound the expected number of samples, we consider two cases (recall that the asymptotics here are taken, unless specified otherwise, while viewing $\p$ as a sequence of distributions $(\p^{(n)})_{n\geq 0}$ and letting $n\to\infty$):
\begin{itemize}
  \item if $\infnorm{\p} = \Omega(\twonorm{\p})$ (i.e., $\infnorm{\p} = \Theta(\twonorm{\p})$), then we denote by $i_\infty$ the element such that $\p_{i_\infty} = \infnorm{\p}$. It follows from properties of the negative binomial distribution that the expected number $M_{\infty}$ of draws necessary to see $\ell=\Theta(\sqrt{k})$ different draws of $i_\infty$ (and thus $k=\binom{\ell}{2}$ collisions) is $\Theta(\sqrt{k}/\infnorm{\p})$, so that $\Ex[M] \leq \Ex[M_{\infty}] = O(\frac{1}{\eps^2 \infnorm{\p}})$.
  
  \item on the other hand, if $\infnorm{\p} = o(\twonorm{\p})$, then we can apply Theorem 4 of~\cite{CP:00} (see also~\cite{MathOverflow:17}) to get that $\Ex[M] \sim_{n\to\infty} \frac{C_k}{\twonorm{\p}}$, where $C_k = \binom{k-\frac{1}{2}}{k-1}\sqrt{\frac{\pi}{2}}\sim_{k\to\infty} \sqrt{2k}$. Recalling that $k=\Theta(1/\eps^4)$, we obtain $\Ex[M] = \Theta(\frac{1}{\eps^2\twonorm{\p}})$, as claimed.
\end{itemize}
\end{proof}

Note that the sample complexity of
Algorithm~~\textsc{Estimate-$\lp[2]$-norm} is tight for near-uniform
distributions (at least, in terms of dependency on
$\twonorm{\p}$). Consider a distribution~$\p$ on $n$ elements with
probability values in $\{(1-\delta)/n,(1+\delta)/n\}$ for some
small~$\delta$. Even though $\twonorm{\p}$ can have sufficiently high
$\twonorm{\p}$ and should be distinguished from the uniform
distribution on $n$ elements, there will be no repetition in the
sample until $\Omega(\sqrt{n})=\Omega(1/\twonorm{\p})$ samples are
taken. The following lemma generalizes this argument.
\begin{lemma}\label{lemma:lb:l2:estimation}
  For any distribution $\p$ and $\eps\in(0,1/3)$, estimation of $\twonorms{\p}$ within a
  multiplicative factor of $(1+\eps)$ requires
  $\Omega(1/(\sqrt{\eps}\twonorm{\p}))$ samples from~$\p$.
\end{lemma}
\begin{proof}
  Take any distribution $\p$. We first consider the case $\eps \geq \normtwo{\p}^2$. Fix any element $c\in\N$ such that $\p(c)=0$ (we can assume for simplicity one exists; otherwise, since we can find, for any $\eta >0$, $c\in\N$ such that $\p(c) < \eta$, we can repeat the argument below for an arbitrarily small $\eta$), and let $\gamma \eqdef \frac{\normtwo{\p}+\sqrt{3\eps+(1+3\eps)\normtwo{\p}^2}}{1+\normtwo{\p}^2}$. Then, we define the distribution $\q$ on $\N$ as the mixture
  \[
      \q \eqdef (1-\gamma\normtwo{\p})\p + \gamma\normtwo{\p}\indic{c}
  \]
  which satisfies $\totalvardist{\p}{\q} = \gamma\normtwo{\p}$, and
  \[
      \normtwo{\q}^2 = (1-\gamma\normtwo{\p})^2\normtwo{\p}^2 +\gamma^2\normtwo{\p}^2 = ((1-\gamma\normtwo{\p})^2+\gamma^2)\normtwo{\p}^2
      = (1+3\eps)\normtwo{\p}^2
  \]
  the last equality from our choice of $\gamma$. Since $\eps<1$, any algorithm that estimates the squared $\lp[2]$ norm of an unknown distribution can be used to distinguish between $\p$ and $\q$. However, from the very definition of total variation distance, distinguishing between $\p$ and $\q$ requires $\bigOmega{1/\totalvardist{\p}{\q}}$ samples. Since
  \[
      \gamma \leq \normtwo{\p}+\sqrt{3\eps+2\normtwo{\p}^2} \leq (1+\sqrt{5})\sqrt{\eps}
  \]
  (as $\normtwo{\p}^2\leq \eps$) 
  we get a lower bound of $\bigOmega{\frac{1}{\sqrt{\eps}\normtwo{\p}}}$.
  
  We now turn to the case $\eps < \normtwo{\p}^2$. The construction will be similar, but setting $\gamma \eqdef {3\eps}/{\normtwo{\p}}$, and spreading the $\gamma\normtwo{\p}= 3\eps$ probability uniformly on $m\eqdef \frac{3\eps}{(1-3\eps)\normtwo{\p}^2}$ elements $c_1,\dots,c_m$ outside the support of $\p$, instead of just one. It is straightforward to check that in this case, the distribution $\q$ we defined is such that
  \[
      \normtwo{\q}^2 = (1-3\eps)^2\normtwo{\p}^2 + \frac{9\eps^2}{m}
      = (1-3\eps)\normtwo{\p}^2
  \]
  so again, by the same argument, any algorithm which can approximate $\normtwo{\p}^2$ to $1+\eps$ can be used to distinguish between $\p$ and $\q$, and thus requires 
  $
      \bigOmega{\frac{1}{\gamma\normtwo{\p}}} = \bigOmega{\frac{1}{\sqrt{\eps}\normtwo{\p}}}
  $
  samples.
\end{proof}

\begin{remark}
  We emphasize that the above theorem is on an instance-by-instance basis, and applies to \emph{every} probability distribution $\p$. In contrast, it is not hard to see that for \emph{some} distributions $\p$, a lower bound of $\bigOmega{1/(\normtwo{\p}\eps^2)}$ holds: this follows from instance from~\cite[Theorem 15]{AOST:17}. This latter bound, however, cannot hold for every probability distribution, as one can see e.g. from a (trivial) distribution $\p$ supported on a single element, for which $\lp[2]$-norm estimation can be done with $O(1/\eps)= O(1/(\normtwo{\p}\eps)$ samples. 
\end{remark}


\subsection{Testing Uniformity}
In this section, we present our algorithm for testing uniformity of a
distribution. We first give a brief overview of the algorithm. The
algorithm first estimates the $\lp[2]$ norm of the input distribution
and uses this value to obtain an estimate on the support size of
the distribution. Then, the algorithm tries to distinguish a uniform
distribution from a distribution that is far from any uniform
distribution by using the number of 3-way collisions in a freshly
taken sample set.  For two distributions with the same $\lp[2]$ norm,
where one is a uniform distribution and the other is far from being
uniform, the latter is expected to produce more 3-way collisions in a
large enough sample set. The algorithm keeps taking samples up to a
number based on the support-size estimate and keeps track of the 3-way
collisions in the sample set to decide whether to accept or reject the
input distribution. 

The following lemma formalizes the intuition that if the $\lp[2]$ and
the $\ell_3$ norm of a distribution is close to those of the uniform
distribution on $N$ elements, then the distribution is close to being uniform.

\begin{lemma}\label{lem:smallnormtouni}
  Let $\p$ be a distribution over $\N$ and $N\in\N$ such that
\[
 \frac{1-\eps}N \le \normtwo{\p}^2 \le \frac{1+\eps}N
\]
and
\[
 \threenormc{\p} \le \frac{1+\delta}{N^2},
 \]
for some $0<\eps,\delta<0.04$. Then, the distance of $\p$ to $\U$
can be upper bounded as 
\[
  \totalvardist{\p}{\U} \le 9\sqrt[3]{\delta+3\eps}.
\]
\end{lemma}
\begin{proof}
  Note that the condition on the $\twonorms{\p}$ implies that $\p$ ``ought to be'' 
  distributed roughly uniformly over $N$ elements, or otherwise would deviate significantly enough from uniformity to impact its $\ell_3$ norm. The condition on
  $\threenormc{\p}$ further strengthens how evenly $\p$ is
  distributed, ensuring that this latter case cannot happen. Below we formalize this intuition and, in particular,
  use the conditions on the norms to upper bound the total mass on the
  items that have probability significantly larger than $1/N$.

Let $R$ be a random variable such that $R$ takes value $p_i$ with
probability $p_i$, for each element $i$ in the support set of
$\p$. Then, $\Ex[R] = \sum_{i\in \N} p_i^2 = \twonorms{\p}$, which implies
\[
 \frac{1-\eps}N \le \Ex[R]  \le \frac{1+\eps}N
\]
and
\begin{align*}
\var[R] &= \Ex[R^2] - \Ex[R]^2\\
&= \sum_{i\in \N} p_i^3 - \twonorm{\p}^4\\
&\le \frac{1+\delta}{N^2} - \frac{(1-\eps)^2}{N^2}\\
&\le \frac{\delta + 2\eps}{N^2}.
\end{align*}

We now derive an upper bound on the $\ell_1$ distance
$\totalvardist{\p}{\U}$. We first obtain an upper bound on the total
weight of elements with probability significantly above or below
$\frac1N$. Then, we can proceed to compare the distribution $\p$ to a
uniform distribution with support size close to $N$.

First, we can bound the total probability mass of items $i$ such that
$p_i>\frac{1+\sqrt[3]{\delta+3\eps}}N$ or
$p_i<\frac{1-\sqrt[3]{\delta+3\eps}}N$ by looking at the probability
of a large deviation of $R$ from its expectation. In particular,
\begin{align*}
\Pr\left[ \left(R > \frac{1+\sqrt[3]{\delta+3\eps}}N\right) \vee
\left(R < \frac{1-\sqrt[3]{\delta+3\eps}}N\right)\right] 
&\le \Pr\left[ | R -\Ex[R] | >
   \frac{\sqrt[3]{\delta+3\eps}-\epsilon}N \right]\\ 
&\le \Pr\left[ | R -\Ex[R] | >
   \frac{\sqrt[3]{\delta+2\eps}}N \right]\\ 
&\le \frac{\var[R]\cdot N^2}{\sqrt[3]{(\delta+2\eps)^2}}\\
&\le \sqrt[3]{\delta+2\eps}
\end{align*}
Note that the second inequality above follows from that
$\sqrt[3]{\delta+3\eps}-\epsilon \ge \sqrt[3]{\delta+2\eps}$ when
$\delta+2\epsilon \le 3^{-3/2} \le 0.18$, by the concavity of the
function $f(x)=\sqrt[3]{x}$ and $f'(x)\ge 1$ for $x\le 3^{-3/2}$.

We now have established that a probability mass of at least
$1-\sqrt[3]{\delta+2\eps}$ of $\p$ is placed on elements with
individual probabilities in the interval 
$[\frac{1-\sqrt[3]{\delta+3\eps}}N,\frac{1+\sqrt[3]{\delta+3\eps}}N]$.
Call this set $F$. Thus, we have that
\[
 \frac{(1-\sqrt[3]{\delta+2\eps})N}{1+\sqrt[3]{\delta+3\eps}}
\le |F| \le \frac{N}{1-\sqrt[3]{\delta+3\eps}}.
\]
Now consider the uniform distribution $\unif_F$ on the set $F$. Since
$\totalvardist{\p}{\U} \le \totalvardist{\p}{\unif_F}$, it suffices to upper
bound the latter. Given that
\[
1-\sqrt[3]{\delta+3\eps} <
1-\sqrt[3]{\delta+2\eps} <
1+\sqrt[3]{\delta+2\eps} <
\frac{1-\sqrt[3]{\delta+3\eps}}{1-\sqrt[3]{\delta+2\eps}},
\]
for any $i\in F$, we have that 
\[
|p_i-\frac{1}{|F|}| \le \frac{4\sqrt[3]{\delta+3\eps}}N.
\]
Finally, we can conclude that
\begin{align*}
\totalvardist{\p}{\unif_F} &= \p(\N\setminus F) + \sum_{i\in F} |p_i -
                 \frac{1}{|F|}|\\
&\le \sqrt[3]{\delta+2\eps} + \sum_{i\in F} \frac{4\sqrt[3]{\delta+3\eps}}N\\
&\le \sqrt[3]{\delta+2\eps} +
  \frac{4\sqrt[3]{\delta+3\eps}}{1-\sqrt[3]{\delta+3\eps}}\\ 
&\le 9\sqrt[3]{\delta+3\eps}
\end{align*}
establishing the lemma.
\end{proof}

The algorithm for testing uniformity is presented below in~\autoref{alg:uniformity}.
 
\begin{algorithm}[H]
  \caption{Testing Uniformity}
  \label{alg:uniformity}
  \begin{algorithmic}[1]
  \Procedure{Test-Uniformity}{$\p,\eps$}
  \State $\delta\gets \eps^3/5832$
  \State $N\gets 1/$\Call{Estimate-$\lp[2]$-norm}{$\p,\delta$}
  \State $k\gets \lceil \eps^{-18}\rceil $
  \State Keep taking samples from $\p$ until you see $k$ 3-way collisions
  or  reach 
  \Statex \hskip\algorithmicindent 
$M=\sqrt[3]{3(1-4\delta)k}N^{2/3}$  samples, whichever happens first. 
  \If{more than $k$ 3-way collisions are observed in the sample set}
  \State \Return \reject
\Else
\State \Return \accept
\EndIf
\EndProcedure
\end{algorithmic}
\end{algorithm}

Note that, for a uniform distribution, $\lp[2]$ norm estimation will
give a reliable estimate $N$ for the support size. Then, we will show
that $M=O(\eps^{-6}N^{2/3})$ samples will be unlikely to produce more
than $k$ 3-way collision. On the other hand, for a distribution that
is far from a uniform distribution, the support size estimation in the
algorithm will be an underestimation. In additions, the $\ell_3$ norm
of such a distribution will be higher than that of the uniform
distribution with that estimated support size. As a result, the
algorithm will observe more than $k$ 3-way collisions in the
subsequent samples with high probability as an evidence that the input
distribution is not uniform.

\begin{theorem} \label{thm:upperbound}
  Algorithm~\textsc{Test-Uniformity}, given independent samples from a
  distribution $\p$ over $\N$ and $0<\eps<\frac12$, accepts if
  $\p\in\U$ and rejects $\p$ such that $\Delta(\p,\U)\ge \eps$,
  with probability at least 3/4. The sample complexity of the
  algorithm is $\Theta(1/\eps^6\threenorm{\p})$.
\end{theorem}
\begin{proof}
In the proof, we will need simple distributional properties of the
number of 3-way collisions, analogous to the arguments in the proof of~\autoref{lem:normest}. Let $T_m$ be the total number of 3-way
collisions in $m$ samples from a distribution $\p$. Then, we have that
\[
 \Ex[T_m] = \binom{m}3\cdot \threenormc{\p}
 \]
and
\[
 \var[T_m] \le O\left(m^3\threenormc{\p} +
m^4\threenorm{\p}^4+m^5\threenorm{\p}^5\right).
\]

For the completeness argument, take $\p=U_S$ for some subset $S$ of
$\N$. Then, by~\autoref{lem:normest}, variable~$N$ from the
algorithm will be within $(1\mp\delta)$ of $|S|$, with probability
3/4. Then, the probability that the number of 3-way collisions in
$m=M$ samples from $\p$ is more than $k$ is
  \begin{align*}
\Pr\left[T_m>k\right] &\le \Pr\left[T_m - \Ex[T_m] > k -
                        (1-4\delta)kN^2\cdot \frac{1}{|S|^2} \right]\\
&\le \Pr\left[T_m - \Ex[T_m] > k - (1-4\delta)kN^2\cdot
                        \frac{(1+\delta)^2}{N^2} \right]\\
&\le \Pr\left[T_m - \Ex[T_m] > \delta k \right]\\
&\le  \delta^{-2}k^{-2}\cdot \var[T_m]\\
&\le  O(\eps^{-6}k^{-2})\cdot O(k^{5/3})\\
&\le \frac{1}{O(\eps^6k^{1/3})}\\
&\le \frac18.
\end{align*}
\tmargin{Constant in the variance needs to be absorbed}

Hence, with constant probability, there will be at most $k$
3-way collisions in the samples from $\p$ and it will be accepted. The
sample and running time complexity is then
\[
\Theta\left(\frac{1}{\eps^6\twonorm{\p}}+\eps^{-6}N^{2/3}\right) =
\Theta\left(\frac{1}{\eps^6\twonorm{\p}}
+\frac{1}{\eps^6\threenorm{\p}}\right)  
= \Theta\left(\frac{1}{\eps^6\threenorm{\p}}\right)
.
\] 

Now, for the soundness argument, suppose that after
$m=M$ samples, at most $k$ 3-way collisions are
observed. We can then argue that, with some constant probability,
$\threenormc{\p}$ is less than $\frac{1+5\delta}{N^2}$. If
$\threenormc{\p}>\frac{1+5\delta}{N^2}$, then
\[
 \Ex[T_m] =\binom{m}3\cdot \threenormc{\p} 
> (1-4\delta)kN^2 \cdot
\frac{1+5\delta}{N^2}
\ge (1+\delta/2)k.
\]
Then,
\begin{align*}
\Pr[T_m \le k ] &=   \Pr[|T_m -\Ex[T_m]| \ge \delta k/2 ]\\
&\le \frac{4\var[T_m]}{\delta^2k^2}\\
&\le O\left(
  \frac{4k^{5/3}N^{10/3}(1+5\delta)^{5/3}}{\eps^6k^2N^{10/3}}
\right) \\
&\le O\left(\frac{1}{\eps^{16}k^2}\right)\\
&\le \frac14
\end{align*}

Hence, we have that
\[
 \frac{1-\delta}N \le \twonorms{\p} \le \frac{1+\delta}N
\]
and
\[
\threenormc{\p} \le \frac{1+5\delta}{N^2}.
\]
By~\autoref{lem:smallnormtouni}, we have that $\p$ is within
$9\sqrt[3]{8\delta}=\eps$ of $\U$.

For a distribution $\p$ that is $\eps$-far from uniform, the
algorithm will stop after observing $k$ 3-way collisions with constant
probability. Similar to the arguments above, this will happen when the
number $m$ of samples satisfies
\[
\binom{m}{3} \cdot \threenormc{\p} \approx k. 
\]
Hence, the sample complexity of the algorithm in this case is
\[
\Theta\left(\frac{1}{\eps^6\twonorm{\p}}+\eps^{-6}N^{2/3}\right) =
\Theta\left(\frac{1}{\eps^6\twonorm{\p}}
+\frac{1}{\eps^6\threenorm{\p}}\right)  
= \Theta\left(\frac{1}{\eps^6\threenorm{\p}}\right)
.
\] 
\end{proof}

\section{The Lower Bound}\label{sec:lowerbound}
In this section, we prove our main lower bound, restated below.
\theomainlb*

\begin{proof}
Let $\q\in\distribs{\domain}$ and $\eps\in(0,1]$ be as in the statement of the theorem. To argue that (a permutation of) $\q$ is hard to distinguish from some $\unif\in\class_U$ with few samples (where ``few'' is a function of $\q$ and $\eps$ only), we will rely on the Wishful Thinking Theorem of Valiant~\cite{Valiant:11}. Indeed, this theorem, broadly speaking, ensures that two distributions with moments (nearly) matching are hard to distinguish given only their fingerprints (equivalently, that distinguishing between relabelings of $\q$ and relabelings of $\unif$ is hard). This will be enough to conclude, as $\class_U$ is a symmetric property.

Specifically, we define the two distributions $\p^\yes,\p^\no$ (respectively in $\class_U$ and \eps-far from it) as follows:
  \begin{itemize}
    \item $\p^\no$ is the ``\no-distribution'' imposed to us -- that is, $\p^\no=\q$;
    \item $\p^\yes$ is a uniform distribution on a set $S\subseteq\domain$ of $1/\normtwo{\q}^2$ elements.
  \end{itemize}
(To see why this is a natural choice: the natural ``\yes-distribution'' to consider in order to fool an algorithm is, by the Wishful Thinking Theorem, a distribution that matches as many moments of $\p^\no=\q$ as possible; which, in our case, will mean matching the $\norm{\cdot}_1$, and $\norm{\cdot}_2$ moments. Note that we could try to \emph{approximately} match the third moment, $\norm{\cdot}_3$, as well, but that there is no hope to match it perfectly: if we could do so with a uniform distribution, this by~\autoref{lemma:equal:moments:characterize:uniform} would imply that $\q$ was in $\class_U$ to begin with.)

In what follows, in view of deriving our lower bound we suppose that $k\norm{\q}_3 \ll 1$. Let $\p^\yes$ be a uniform distribution on a subset of $m\eqdef \frac{1}{\normtwo{\q}^2}$ elements. Computing the $k$-based moments of $\p^\yes$ is straightforward: for any $j\geq 2$, we have
\[
      m^\yes(j) = \frac{k^j}{m^{j-1}} = k^j \normtwo{\q}^{2(j-1)} = \frac{\left(k \normtwo{\q}^{2}\right)^j }{\normtwo{\q}^{2}}
\]
while, of course, $m^\no(j) = k^j \norm{\q}_j^j$. It follows that\cmargin{Keep in mind that we need $k \leq \norminf{\q}/500$.}
  \begin{align*}
     \sum_{j=2}^\infty \frac{\abs{m^{\yes}(j)-m^{\no}(j)}}{\sqrt{1+\max(m^{\yes}(j),m^{\no}(j))}}
     &= \sum_{j=3}^\infty \frac{\abs{m^{\yes}(j)-m^{\no}(j)}}{\sqrt{1+\max(m^{\yes}(j),m^{\no}(j))}} \\
     &= \sum_{j=3}^\infty k^j \frac{\abs{\norm{\q}_j^j - \normtwo{\q}^{2(j-1)}}}{\sqrt{1+k^j \max(\normtwo{\q}^{2(j-1)},\norm{\q}_j^j)}}\;.
  \end{align*}
  Now, we will use~\autoref{fact:ineq:moments:holder} to get rid of the absolute value; as it enables us to rewrite our sum as
  \begin{align*}
     \sum_{j=2}^\infty \frac{\abs{m^{\yes}(j)-m^{\no}(j)}}{\sqrt{1+\max(m^{\yes}(j),m^{\no}(j))}}
     &= \sum_{j=3}^\infty k^j \frac{\norm{\q}_j^j - \normtwo{\q}^{2(j-1)}}{\sqrt{1+k^j \norm{\q}_j^j}}\;.
  \end{align*}
  In order to handle this last expression, we can drop the denominator, to get
  \begin{align*}
     \sum_{j=2}^\infty \frac{\abs{m^{\yes}(j)-m^{\no}(j)}}{\sqrt{1+\max(m^{\yes}(j),m^{\no}(j))}}
     &\leq \sum_{j=3}^\infty k^j \left( \norm{\q}_j^j - \normtwo{\q}^{2(j-1)} \right)
     \leq \sum_{j=3}^\infty k^j\norm{\q}_j^j\\
     &\leq \sum_{j=3}^\infty k^j\norm{\q}_3^j\tag{Monotonicity of $\lp[p]$ norms} \\
     &= \frac{k^3\norm{\q}_3^3}{1-k\norm{\q}_3} < \frac{1}{24}
  \end{align*}
  using our assumption that $k\norm{\q}_3 \ll 1$.
  
  This last bound will allow us to apply~\autoref{theo:valiant:wishful} and obtain the lower bound, provided that $\norminf{\p^\yes},\norminf{\p^\no} \leq \frac{1}{500k}$. But this last condition follows from observing that $k\max(\norminf{\p^\yes},\norminf{\p^\no})\leq k\max(\norm{\p^\yes}_3,\norm{\p^\no}_3) = k\max(\norm{\q}_3,\norm{\q}_2^{4/3}) \leq k\norm{\q}_3 \ll 1$.
  
  \end{proof}

\begin{remark}
  Although our lower bound does not directly feature a dependence on the distance parameter $\eps$ (besides applying to any $\eps' \leq \eps$), we conjecture that the right dependence should be linear in $1/\eps$, i.e., $\bigOmega{1/(\eps\norm{\q}_3)}$. (Indeed, while a \emph{square} dependence on $\eps$ appears natural, it cannot hold on an instance-by-instance basis for \emph{all} distributions, analogously to that of~\autoref{lemma:lb:l2:estimation}: as one could see by considering a degenerate distribution $\q$ with $1-\eps$ probability weight on a single element, for which  uniformity testing can be done with $O(1/\eps)=\bigOmega{1/(\eps\norm{\q}_3)}$ samples.) Establishing this linear dependence with our techniques, however, would require at the very least a significant strengthening of the above chain of inequalities, especially at step $(\dagger)$.
\end{remark}
  

\clearpage
\bibliographystyle{alpha}
\bibliography{references} 

\clearpage
\appendix

\end{document}